\documentclass{article}
\usepackage[utf8]{inputenc}
\usepackage{amsmath,amssymb,amsthm,graphicx}
\usepackage{url}
\usepackage{geometry}

\newcommand{\vektor}[1]{\begin{pmatrix} #1 \end{pmatrix}}
\newcommand{\B}{\mathcal{B}}
\newcommand{\R}{\mathbb{R}}
\newcommand{\pil}{\rightarrow}
\newcommand{\de}{\: \mathrm{d}}
\newcommand{\bt}{\mathbf{T}}

\newtheorem{corr}{Corollary}
\newtheorem{defi}{Definition}
\newtheorem{eks}{Example}
\newtheorem{prop}{Proposition}
\newtheorem{rem}{Remark}

\author{Niels Lundtorp Olsen\footnote{\textit{Department of Applied Mathematics and Computer Science}, {Technical University of Denmark}}, Alessia Pini\footnote{\textit{Department of Statistical Sciences}, {Università Cattolica del Sacro Cuore}}, Simone Vantini\footnote{\textit{Department of Mathematics}, {Politecnico di Milano}}}
\title{Local inference for functional data on manifold domains using permutation tests}

\begin{document}

\maketitle

\begin{abstract}
Pini and Vantini (2017) introduced the interval-wise testing procedure which performs local inference for functional data defined on an interval domain, where the output is an adjusted p-value function that controls for type I errors. 

We extend this idea to a general setting where domain is a Riemannian manifolds. This requires new methodology such as how to define adjustment sets on product manifolds and how to approximate the test statistic when the domain has non-zero curvature.
We propose to use permutation tests for inference and apply the procedure in three settings: a simulation on a "chameleon-shaped" manifold and two applications related to climate change where the manifolds are a complex subset of $S^2$ and $S^2 \times S^1$, respectively. 
We note the tradeoff between type I and type II errors: increasing the adjustment set reduces the type I error but also results in smaller areas of significance. However, some areas still remain significant even at maximal adjustment. 

\end{abstract}

\section{Introduction}
A topic of increasing importance in the functional data analysis (FDA) literature is \emph{local inference}, i.e., for a given null hypothesis testing problem assessing in which parts of the functional domain $D$ a given null hypothesis is true/false. More formally, if we let $H_0^t$ be a null hypothesis defined for each point $t \in D$, then we wish to infer the set $\{t \in D: H_0 \text{ is false/true}\}$. The crucial issue is that of \emph{multiple testing}: we are indeed in the scenario of simultaneously testing a continuous infinity of hypotheses. So, if the hypotheses are evaluated separately, any control over the domain of the family-wise error rate (FWER) would be completely lost thus leading to a useless tool from an applicative point of view. 

A number of methods for local inference for functional data have been proposed in recent literature, e.g.
 \cite{pv2017interval}, \cite{olsen2021}, \cite{winkler2014}.

At the heart of the issue is the trade-off between doing type I errors and the power of the test  (ie. type II errors). 
Most methods focus on limiting the \emph{family-wise error rate} (FWER), either globally (\cite{winkler2014}, \cite{liebl2019}),
 or across an interval \cite{pv2017interval}. 
A popular and less conservative measure is to control the \emph{false discovery rate} (FDR) \cite{BH1995} which was studied in an FDA setting by \cite{olsen2021}. 
Furthermore, we can separate into \emph{global adjustment procedures} and \emph{local adjustment procedures}. Local adjustment procedures make use of 'proximity' in the adjustment procedure (e.g. (only) adjusts for intervals below a given length $L$), whereas global adjustment procedures adjust wrt. the entire functional domain and no proper subsets thereof.
The IWT procedure \cite{pv2017interval} is the only published local adjustment procedure that we know of; the rest are global procedures.
These methods have been compared in \cite{pataky2021}. 

\paragraph{Local inference in higher dimensions}
Most research/methodology in local inference has focused on one-dimensional domains, i.e. intervals. 
Some methods are directly applicable to higher dimensions \cite{winkler2014}, \cite{olsen2021}, whereas others would require some extension of the methodology. We notice that the aforementioned methods all use a global adjustment procedure, pointing to a gap in the existing literature. 

There is marked difference between local adjustment in the one-dimensional and multidimensional cases:
the only shape in 1D is the interval, whereas higher dimensions contains an infinite multitude of shapes. 
Furthermore, the notion of proximity must deal with dimensional inhomogenity when the dimensions/coordinates represent incomparable quantities, for instance space and time.

In this work we extend the IWT methodology to Riemannian manifolds of any dimension. In a similar fashion, we use permutation tests to construct $p$-values on a family of pre-specified 'adjustment sets', which are then adjusted to obtain an adjusted $p$-value function. 
Our proposed method has a large deal of flexibility built in; we include an "upper limit" (potentially $\infty$) on the adjustment and allow for incomparable quantities.  

As with any methodology for functional data, one needs to use certain approximations, and
there is a non-trivial issue of how to take curvature etc. into consideration when approximating the manifold. 
\cite{olsen2021} proposes to map the manifold into some set $T \subseteq \R^k$ and use a corresponding weighted measure. Here we propose to use a triangulation, see Section \ref{implement-sect}. 

The remainder of the paper is organized as follows: 
Section \ref{sect-model} defines the statistical model. Section \ref{hyp-section} defines the inferential procedure, and Section \ref{implement-sect} describes and discusses the considerable details needed for an actual implementation. Sections \ref{sect-sim} and \ref{sect-application} contains simulations and application studies, and finally in Section \ref{sect-discuss} we highlight and discuss important points of this article.

\section{Setting and model} \label{sect-model}
In this section we describe the setting of our statistical model. The setting of our models and data are functional data defined on a \emph{Riemannian manifold}, which is considered fixed and known.

The setting of our models and data are \emph{Riemannian manifolds}, with the intrinsic metric, topology and measure/density derived from this manifold. 

\begin{defi} \label{manif-defi}
A manifold $M$ of finite dimension is a smooth manifold together with a smoothly varying 2-tensor field $g$ on $M$ which is an inner product at each point. %
The inner product $g$ defines a metric $d$ and a measure $\mu$ on $M$, which we will refer to as the \emph{Riemannian metric} and the \emph{Riemannian measure}, respectively. 
\end{defi}

For a given  manifold $M$ with metric $d$, define $B(x, \epsilon)$ as 
$$
B(x,\epsilon) = \{y \in M | d(x,y) < \epsilon\}, \quad x \in M, \epsilon > 0 
$$

For simplicity, we restrict ourselves to non-weighted Riemannian measures, but note the (straightforward) option of introducing a weight into the measure.

\subsection{Statistical model and pointwise hypotheses}
Let $M_1, \dots, M_L$ be  manifolds according to Definition \ref{manif-defi}, and define $M = M_1 \times \dots \times M_L$. The structure of $M$ is not of relevance in the modelling, but will be of interest in the inference. 

We assume that we have observed $N$ smooth functional signals $\xi_1, \dots, \xi_N$: $M \pil \R$ that are generated according to the following functional linear model:
\begin{equation}
\xi_i(s) = \beta_0(s) + \sum_{k = 1}^K \beta_k(s) x_{ki}(s) + \epsilon_i(s), \quad s \in M, \quad i = 1, \dots, N
\end{equation}
Here $x_{i1}, \dots , x_{Ki}$ are known  covariate functions, and $\beta_{1}, \dots , \beta_{K}$ are the unknown functional regression functions. The error signals $\epsilon_1, \dots, \epsilon_N$ are i.i.d. continuous zero-mean random functions. 
All functions are assumed to be continous, which also ensures measurability. 

Model estimation is commonly done using ordinary least squares estimation (OLS). However, in this work we are not interested in the estimation, %

This could in principle be almost any pointwise defined hypothesis, however following \cite{abramowiczmox} we will restrict our to  affine  (possibly composite) hypotheses on the form $C \beta = c_0$ for a fixed matrix $C$; this covers all common hypotheses of the general linear model.

In detail, let $C \in \R^{m \times K}$ be a full-rank matrix and $c_0: M \pil \R^m$, such that this pair defines a pointwise hypothesis for each $s \in M$:
\begin{align}
    H^s_0: & \quad  C\beta(s) = c_0(s)  \\
    H^s_A: & \quad C\beta(s) \neq c_0(s) 
\end{align}

\begin{eks} \label{two-sample-model} 
Consider spatial data defined on a region $M = [0,X] \times [0, Y] \subset \R^2$. The (unpaired) two-sample model can be defined as:
\begin{align*}
\xi_{Ai}(t) =& \beta_A(t) + \epsilon_{Ai}(t), \quad i = 1, \dots, n_1, \quad t \in M  \\
\xi_{Bj}(t) =& \beta_B(t) + \epsilon_{Bj}(t), \quad j = 1, \dots, n_2, \quad t \in M
\end{align*}
Here we estimate the population means $\beta_A$ and $\beta_B$ by pointwise averages, ie. 
\begin{align*}
\hat{\beta}_A(t) = \tfrac{1}{n_1} \sum_{k = 1}^{n_1} \xi_{Ai}(t) \quad, t \in M 
\end{align*}
and similar for $\beta_B$. The usual statistical question to ask is then: \emph{is $\beta_A = \beta_B$?}, and in the case of local inference: \emph{for which $t \in M$ are  $\beta_A$ and $\beta_B$ equal?}
\end{eks}

\section{Methodology: local inference and adjustment procedure} \label{hyp-section}

This section formally introduces the adjustment sets, the testing procedure and the adjustment procdure. Details of implementation and approximation are described in Section \ref{implement-sect}. 

\subsection{Adjustment sets} First we need to define an \emph{adjustment family}, a family of sets that are used to adjust the $p$ value function and control the associated error rate.

In 1D, there is one natural adjustment family: namely all intervals, as used in \cite{pv2017interval}. 
For a general manifold $M$ we use the following approach:
Let  $M = M_1 \times \dots \times M_L$ and let $r_1, \dots r_L > 0$ be corresponding radii, defined \textit{a priori}.  We allow for $r_l = \infty$.

The adjustment family $\B$ consists of all sets $I$ on the following form ("ball products on M"): 
\begin{equation}
I = B(x_1, \epsilon_1) \times \dots \times B(x_1, \epsilon_n), \quad x_i \in M_i, \quad \epsilon_i \leq r_i \text{ for all } i \label{iwt-ball}
\end{equation}

\begin{rem}
	By definition, the Cartesian product of finitely many manifolds is a manifold itself. However, we wish to make the distinction between e.g. the cases $\R \times \R$ and $\R^2$, since the product of two balls in $\R$ is not a ball in $\R^2$.
	
	From a practical/application point of view, $\R^2$ could represent e.g. a 2D spatial domain, whereas $\R \times \R$ would represent two incomparable quantities, e.g. time and space. See also Example \ref{two-sample-model-2}. 
\end{rem}

\begin{rem} The radii used to define $\B$ are an addition to the previous methodology of \cite{pv2017interval} which only considered the case $r = \infty$.
	Besides allowing the researcher more flexibility in terms of inference, the choice of $r$ also illustrates the link between type I and II errors: decreasing $r_i$ increases the power of the testing procedure since we need to correct for fewer sets. We illustrate this trade-off between type I and II errors in simulation and applications. 
\end{rem}

\subsection{Ball-wise hypotheses, inference and $p$-values}

Following the notation introduced in the previous section, we define the ball(-wise) (null) hypothesis and alternative hypothesis for all $I \in \B$ as:

\begin{align}
    H^I_0: & \quad  C\beta(s) = c_0(s) \quad \forall s \in I \\
    H^I_A: & \quad C\beta(s) \neq c_0(s) \textbf{ for some } s \in I
\end{align}
noting that $H_0^I$ is { true} if and only if the pointwise null hypothesis $H^s_0$ is{ true} for all $s \in I$.

\begin{eks}[Continuation of example \ref{two-sample-model}] \label{two-sample-model-2}
The hypothesis of  equality in means is given by $\beta_A = \beta_B$, or, using the matrix notation by $C = \vektor{-1 & 1}$ and $\beta = \vektor{\beta_A & \beta_B}^\top$.

We now have two choices: we can either view $M$ as a submanifold of $\R^2$ (corresponding to an 'actual' spatial domain) or as the product of $[0,X]$ and $[0,Y]$.
In the former case, ballwise hypotheses are circles:
$$
H_0^I: \beta_A(t) = \beta_B(t), \quad |t-x| < r, \quad t \in M
$$
for $r > 0, x\in M$, whereas in the latter case we get rectangles:
$$
H_0^I: \beta_A(u,v) = \beta_B(u,v), \quad x_1 < u < x_2, \quad y_1 < v < y_2
$$
where $0 \leq x_1 < x_2 \leq X$ and $0 \leq y_1 < y_2 \leq Y$. 

An illustration of this difference can be seen in Figure \ref{fig-h0-example-ill}. 
\begin{figure}
    \centering
    \includegraphics[width=\textwidth]{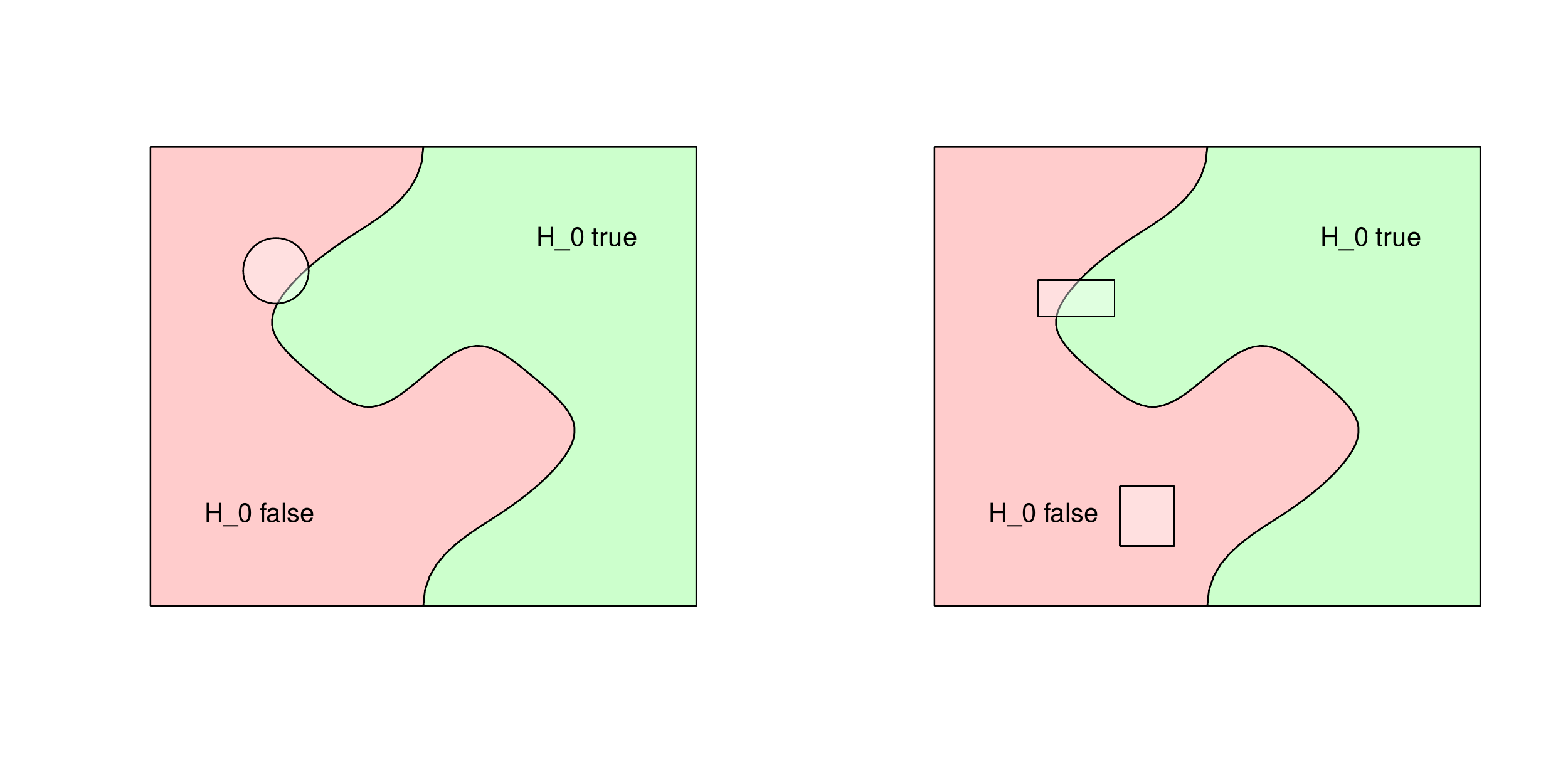}
    \caption{A 2D spatial domain (left) vs. a product domain (right). In the former case, the adjustment sets are circles. In the latter case,   the two domains are considered separate quantities, and adjustment sets are rectangles. 
    }
    \label{fig-h0-example-ill}
\end{figure}
\end{eks}

\paragraph{Inference and $p$-values}
We propose to test hypotheses using a test statistic $t: \R^N \pil \R_+$. We assume the  resulting map $T: M \pil \R_+, T(s) = t(\xi_1(s), \dots, \xi_N(s))$ to be continuous on $M$ almost surely. 

We define the ball-wise test statistic $T^I$ as
$$
T^I = \int_I T(x) \de \mu(x)
$$
and let point-wise and ball-wise $p$-values be defined in the usual sense as the probability of observing a more extreme event under $H_0$:
\begin{equation}
    p(x) = P_{H_0}(T(x) > T_{\text{obs}}(x) ) \quad x \in M
    \label{p-unad-eq}
\end{equation}
and
$$
p^I = P_{H_0}(T^I > T^I_{\text{obs}})
\quad I \in \mathcal{B}
$$
Under weak conditions it holds that  $p(x) = \lim_{I \pil x}$ for almost all $x$, where  the limit shall be understood as balls containing $x$ of decreasing radii. 
We will refer to the function $p$ (from Eq \eqref{p-unad-eq}) as the \emph{unadjusted $p$-value function}. 

\paragraph{Adjusted $p$-value function and control of error rate} 

Following the approach of \cite{pv2017interval} we define the \emph{adjusted $p$-value function} as 
$$
\tilde{p}(x) = \sup_{I \in \mathcal{B}: x \in I} p^I \quad \text{for } x \in M
$$

Let $U \subseteq M$ be the set of  points where $H^0$ is true. The pointwise, ball-wise and family-wise error rates are defined in the usual sense as the chance of committing a type I error on a given significance level $\alpha$:
\begin{align*}
    \text{pointwise:} \quad & P(p(x) \leq \alpha) \quad x \in U \\
        \text{ball-wise:} \quad & P(\exists {x \in I}: \tilde{p}(x) \leq \alpha) \quad I  \in \mathcal{B}, I \subseteq U   \\
        \text{family-wise:} \quad & P(\exists {x \in U}: \tilde{p}(x) \leq \alpha)
\end{align*}

\begin{prop}
Let $p$ and $\tilde{p}$ be the unadjusted and adjusted $p$-value functions, respectively. Then $p(x) \leq \tilde{p}(x)$ and, importantly, $\tilde{p}$ controls the ball-wise error rate:
$$
P (\exists {x \in I}: \tilde{p}(x) \leq \alpha) \leq \alpha
$$
(or)
$$
P ( \tilde{p}(x) \leq \alpha \text{ for all } x \in I) \leq \alpha
$$
\end{prop}

\begin{proof}
$\tilde{p}(x) \leq p(x)$ is trivial. Let $I \in \mathcal{B}$ be given. If we assume $H^t_0$ is true for all $t \in I$, then $P(p^I \leq \alpha) \leq \alpha$. But $\tilde{p}(x)  \geq p^I$ for all $x \in I$, and so $P (\forall {x \in I}: \tilde{p}(x) \leq \alpha) \leq \alpha$.
\end{proof}

\begin{corr}
If the set of points where $H^0$ is true can be characterised by $I$ for some $I \in \mathcal{B}$, then the adjusted $p$-value function $\tilde{p}$ controls the family-wise error rate. 
\end{corr}

\section{Approximation and implementation} \label{implement-sect}
Though the definitions presented in the previous section are mathematically straightforward, actual evaluation of the adjusted and unadjusted $p$-value functions poses a challenge, and has to rely on an approximation, as with any methodology for functional data.

In the one-dimensional case, ie. when $M = [a,b] \subset \R$, evaluating $T^I$ is relatively straightforward. 
It has been studied in several papers starting with \cite{pv2017interval} who used an equidistant grid for approximating $T^I$. 

For manifolds $M$ which are products of intervals $[a_i, b_i]$, generalizing the disrectization procedure of \cite{pv2017interval} is also straightforward, but has not previously been published.

However, the discritization and evaluation becomes challenging when $M$ involves multidimensional manifolds, in particular if the curvature is non-zero.
\cite{olsen2021} proposes to map the manifold into some set $T \subseteq \R^k$ and use a corresponding weighted measure. This solves the issue of curvature, but not that of discretization.

The focus of our proposal is on the case of non-zero curvature, in general considering the dimension of $M_i$ is low (ie. two or three). Most applications would fall into this category as physical space only has three spatial dimensions. 

\subsection{Approximating the test statistic $T^I$} \label{sec-approx-test}
We propose to use triangulation to approximate test statistic values $T^I$. First note that for $I \in \mathcal{B}$, $I$ is a Cartesian product of balls. By Fubini's theorem, we can expand  integrals on $I$ as:
$$
\int_I f(u) \de u = 
\int_{B(x_1, \epsilon_1)} \dots \int_{B(x_n, \epsilon_n)} f(u) \de u_n \dots \de u_1
$$
where $I$ here is on the form \eqref{iwt-ball}.
Hence we can restrict our focus to "single" manifolds and use recursion to obtain the final result $T^I$. For the remainder of this section we will therefore assume that $M$ is the product of a single manifold, $M = M_1$ of dimension $k$.

Let $x \in M$ and $r > 0$ be given. We evaluate $\int_I T(u) \de u$ by the following procedure: 

\begin{enumerate}
    \item A set of points $ E = \{e_1, \dots, e_m\} \in M$ are selected.
    \item A triangulation $\mathbf{T}$ from $E$ is constructed, where each element of $\mathbf{T}$ is a $k$-simplex. 
    \item The integral is approximated from $\mathbf{T}$ in algorithm described below. 
\end{enumerate}
By default, one would use the same triangulation for all $I \in \mathcal{B}$ and all realisations of $T$. In the followin, we describe how to deal with points 1, 2, and 3.

\paragraph{Point selection and triangulation} There exist a vast literature for triangulations of a manifold, in particular when $M$ is a 2D surface in $\R^3$ due to its use in computer graphics. 
Procedures for triangulation is beyond the scope of this article, though we note that properties such as Delaunay triangulation and a roughly uniform grid are reasonable assumptions.

\paragraph{Integral approximation:} For simplicity, we only consider here the case $k= 2$. It  is easily extended to the case $k > 2$.

Let $(E, \bt)$ denote a triangulation of $M$, where $E$ is the set of vertices corresponding to $\bt$.
For $S \in \mathbf{T}$, let $v(S)$ denote the vertices of $S$ in $E$, this is a set of size $3$. 

We define the {area} $A(S)$ as the area of the (euclidean) triangle $S$ spanned by $v(S)$, using the geodesic distances. If $S$ is sufficiently small, then $\sigma(S)  \approx A(S)$ where $\sigma(S)$ is the measure induced by the manifold. 
Next, we use $A(S)$ to weight the vertices in $E$. In detail, define
\begin{equation}
W(e) = \frac{1}{3} \sum_{S: e \text{ is a vertex of }  S} A(S), \quad e \in E
\label{weight-eq}    
\end{equation}

Then, for $f: M \pil \R$, we approximate $\int_{B(e, r)} f(x) \de x$ by:
\begin{equation}
\int_{B(x, r)} f(u) \de u = \sum_{\{v \in E: d(x,e) < r \} } W(e) f(e) \label{integral-approx-eq}
\end{equation}
 For a sphere, the relative error between $\sigma(S)$ and $A(S)$ is quadratic  in $(a/R)^2$, where $a$ is a triangle side length and $R$ is the radius of the sphere\footnote{If we consider a right spherical triangle with sides $(a,a, \arccos(\cos^2 a))$, then $\sigma(S) = \frac{1}{2} a^2 + \frac{1}{12} a^4 + o(a^5)$ and $A(S) = \frac{1}{2} a^2 + o(a^5)$ }, and the error between $f(u)$ and $f(e)$ can be controlled if we for instance assume a Lipschitz condition on $f$. 
 A more detailed error assessment depends intrinsically on the geometry of the surface and the triangulation, and is left for future studies.

\subsection{Inference of the test statistic under $H_0$} \label{sec-test-inf}
We can approximate $T^I$ arbitrarily well using \eqref{integral-approx-eq}, however the distribution of $T^I$ under $H_0$ depends on the distribution of the error signals $\epsilon_k$ and is generally intractable. 

Following \cite{pv2017interval} and \cite{abramowiczmox} we propose to use permutation tests, more specifically the Freedman and Lane permutation scheme \cite{freedman1983nonstochastic}, as defined by \cite{pesarin}. 
If the model and hypothesis is "simple" (ie. data are iid under $H^0$, then the permutation test is exact and controls the $p$-value, otherwise the $p$-value control is only asymptotically exact.

\subsection{Computational considerations}
Any suitably fast implementation would calculate all pairwise distances between points in $E$ up front. However, if the size of $E$ is large (e.g. 100000), the number of pairwise distances is $\approx E^2/2$ and will exceed the memory capacity of a typical laptop. Therefore we do not recommend using more than 10000 points, the simulation and applications used 8450, 6252 and 9695 points, respectively. 

We also note that the testing part (ie. permuting the data and applying the method described in this section) and the adjustment part are separate procedures, and that the 'cost of adjustment' does not increase with increased number of permutations.  

\section{Simulation study} \label{sect-sim}

\subsection{Simulation study 1: "The chameleon"}
To demonstrate our methodology on an unusually shaped manifold, we decided to apply it to a "chameleon figure": a 3D graphic of a chameleon\footnote{Courtesy of ADAPTA Studio (\cite{locatelli1}, \cite{perotto1}).}.
As seen on Figure \ref{fig-cam2}, the chameleon has many details, making it a highly non-trivial manifold. 

\begin{figure}
    \centering
\includegraphics[width=0.24\textwidth,trim = {195 55 195 0},clip]{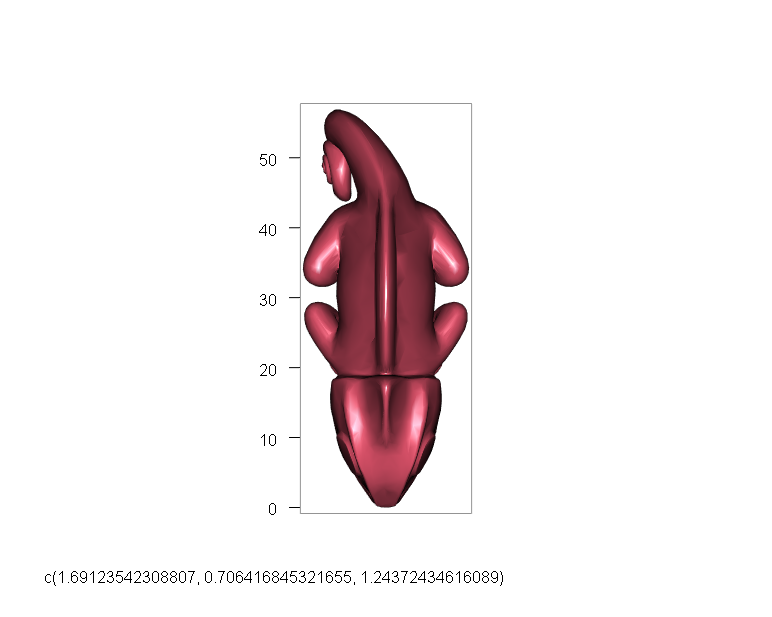}
\includegraphics[width=0.24\textwidth,trim = {195 55 195 0},clip]
{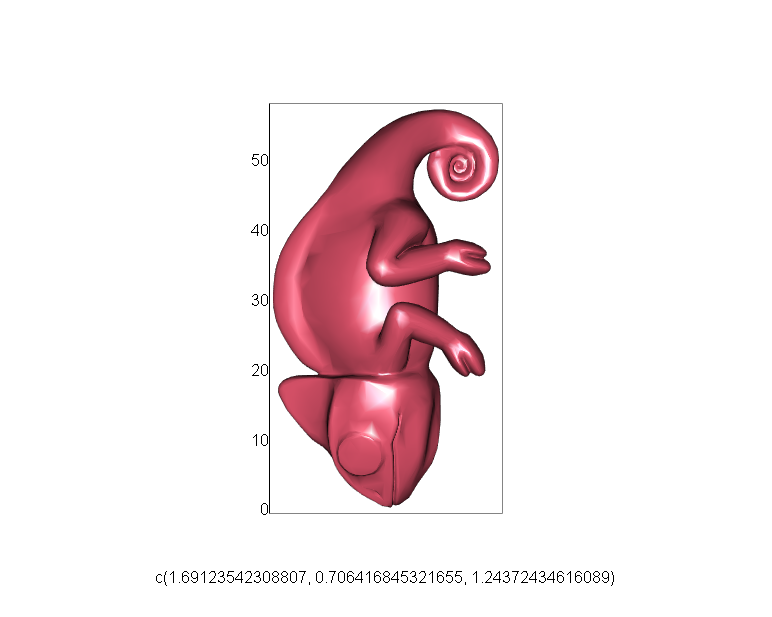}
\includegraphics[width=0.24\textwidth,trim = {195 55 195 0},clip]{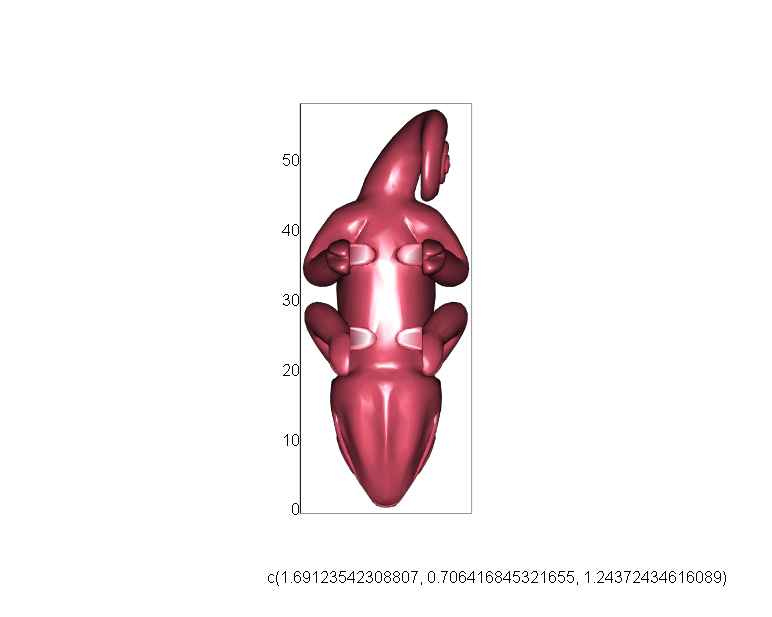}
\includegraphics[width=0.24\textwidth,trim = {195 55 195 0},clip]{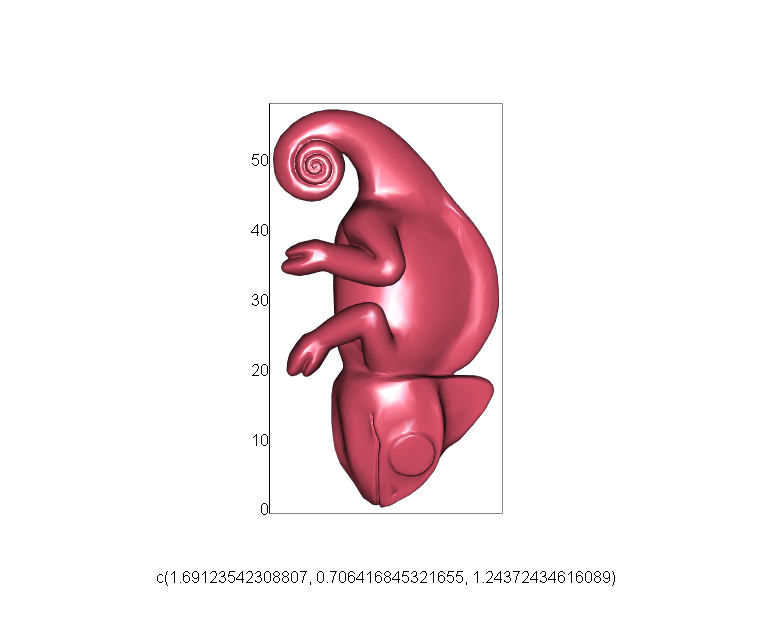}\\ \includegraphics[width=0.25\textwidth,trim = {255 75 290 105},clip]{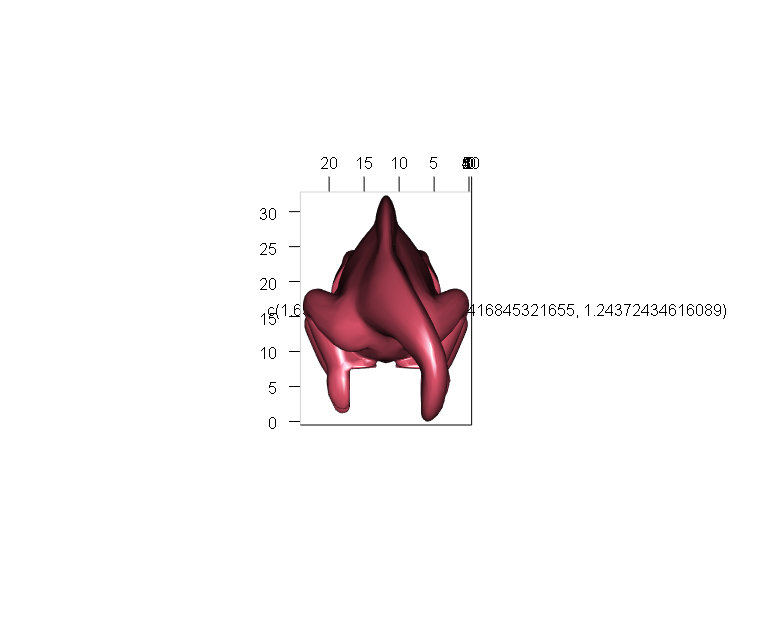}
\includegraphics[width=0.25\textwidth,trim = {205 5 255 165},clip]{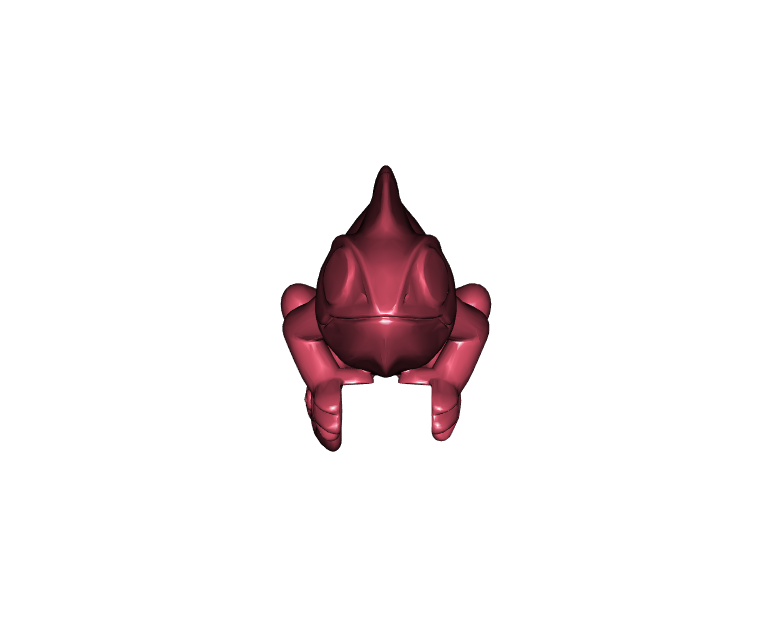}\\
\includegraphics[width=0.49\textwidth]{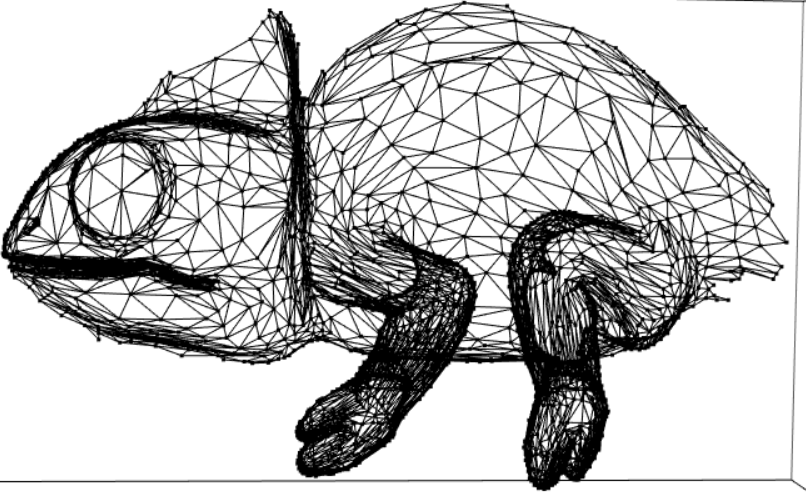}  
    \caption{Views of the chameleon and its triangulation}
    \label{fig-cam2}
\end{figure}

The chameleon figure is represented by a triangulation consisting of 8450 points. We applied this triangulation directly into our procedure, and we used the inherent distance as found by Dijkstra's algorithm.

\paragraph{Model and hypothesis} We simulated data according to the model $y_i = \theta_i + \epsilon_i$, for $i = 1, \dots, N$, where $\theta_i$ is either zero or the base signal $\phi$, and $\epsilon_i$ is an error term generated using a Gaussian kernel. Half of the observations were assigned $\theta = 0$ and the other half $\theta = \phi$. Hence this is a two-sample setup, and we considered the point-wise null hypotheses $H^0_x: \phi(x) = 0$.

As base signal $\phi$ we selected zero for most of the chameleon, and a constant $c$ for the remaining chameleon, this corresponds to the region where $H_0$ is false. 
We selected two different settings for the regions of non-zero signal: the four feet and the right side of the chameleon's 'crest'. The two regions of non-zero signal are illustrated in Figure \ref{figure-cameleon-signal}.
We also varied the number of samples $N$ and the maximal radius $r$ (the length of the cameleon is 57 units, so the extreme case of $r = 0.5$ corresponds to a ball of very small size. In total we have 12 different scenarios; an overview is in given in Table \ref{table-cameleon-scen-overview}.

\begin{figure}
\centering
\includegraphics[width=0.49\textwidth]{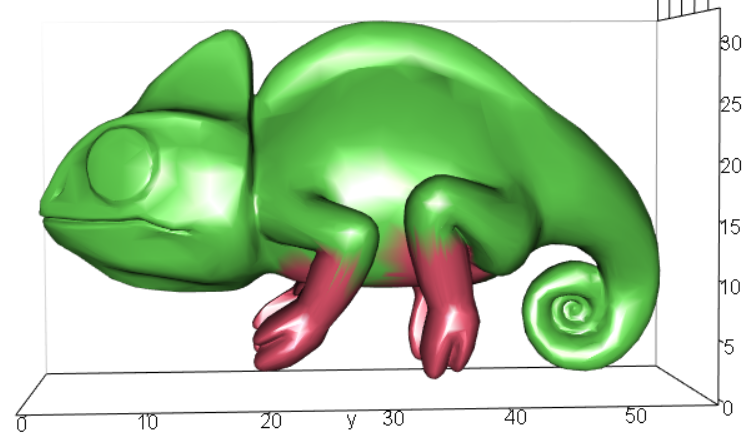}
\includegraphics[width=0.49\textwidth]{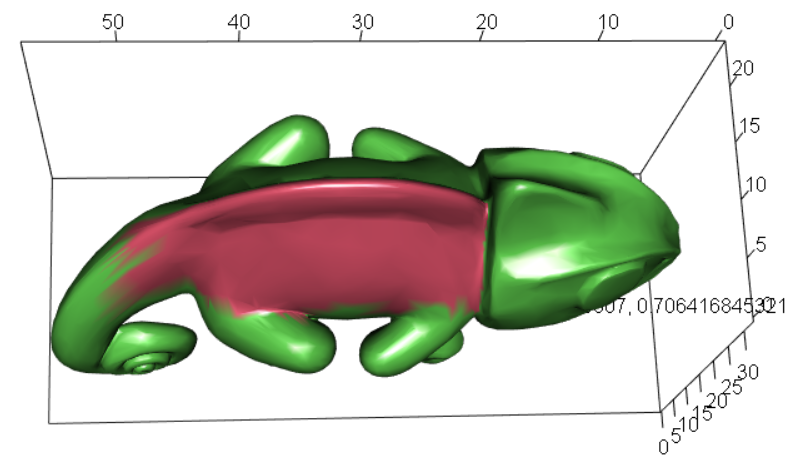}
    \caption{Base signals for cameleon simulation. Left: four feet. Right: right side of crest}
    \label{figure-cameleon-signal}
\end{figure}

\begin{table}[!htb]
    \centering
    \begin{tabular}{cccc}
    Scenario & Region of $H_0$ false  & no. of samples ($N$) & max radius ($r$)\\ \hline
      1 & Four feet & 20 & $\infty$ \\
    2 & Four feet  & 10 & $\infty$  \\ 
    3 &Four feet   & 40 & $\infty$   \\
    4 & Four feet  & 20 & $10$  \\
    5 & Four feet  & 20 & $3$ \\ 
    6 & Four feet  & 20 & $0.5$ 
    \end{tabular}
    \caption{Parameters for the different "feet" scenarios. The six "crest" scenarios have the same parameters for $N$ and $r$.}
    \label{table-cameleon-scen-overview}
\end{table}

\paragraph{Test}
As test statistic we used the pointwise t-test statistic: $T(x) = \left(\frac{\overline{y_1(x)} - \overline{y_2(x)}}{s(x)}\right)^2$, where $\overline{y_1(x)} $, $\overline{y_2(x)}$ and $s(x)$ are the average of the first half of the signals, the average of the second half of the signals and the estimated  standard deviations, respectively, evaluated in $x$. Inference of the test was done according to the procedure outlined in Section \ref{sec-approx-test} and \ref{sec-test-inf}.
We used 500 permutations for the permutation tests, and each scenario was repeated 250 times. All tests were performed on a 5\% significance level. 

\paragraph{Results} 
We calculated the sensitivity (power), family-wise error rate (FWER), false positive rate and false discovery rate for each of the scenarios, results can be found in Table \ref{table-cam-results}. As expected sensitivity increased with  number of samples, whereas FWER, the false positive rate and false discovery rate remained stable. 
Decreasing $r$ gave interesting results: we observe a large large increase in sensitivity, however at the cost of increasing the false discovery rate. 
At the lowest value of $r$, sensitivity reached 1 and the false positive rate reached 0.05 (the significance level), indicating virtually no adjustment. 

\begin{table}
Feet: \\
\begin{tabular}{c|cccc}
Scenario & Sensitivity & FWER & False positive rate & False Discovery Rate  \\ \hline
1 & 0.653 & 0.264 & 0.004 & 0.016 \\
2 & 0.392 & 0.204 & 0.004 & 0.021 \\
3 & 0.815 & 0.216 & 0.004 & 0.012 \\
4 & 0.816 & 0.280 & 0.008 & 0.026 \\
5 & 0.959 & 0.564 & 0.028 & 0.076 \\
6 & 1.000 & 0.696 & 0.049 & 0.119 
\end{tabular} \\
Crest: \\
\begin{tabular}{c|cccc}
Scenario & Sensitivity & FWER & False positive rate & False Discovery Rate  \\ \hline
1 & 0.282  & 0.096 &  0.004 & 0.028 \\
2 & 0.062 & 0.108 & 0.002 & 0.036 \\
3 & 0.535 & 0.128 & 0.005 & 0.028 \\
4 & 0.642 & 0.312 & 0.011 & 0.066 \\
5 & 0.945 & 0.564 & 0.033 & 0.146 \\
6 & 0.999 & 0.648 & 0.053 & 0.207 
\end{tabular}

\caption{Results of simulation on a 5\% significance level. Upper table: "four feet" scenarios. Lower table: "crest" scenarios. }
    \label{table-cam-results}
\end{table}

\section{Applications} \label{sect-application}
We apply the presented methodology to two data sets related to global warming: satellite measurements of seasonal temperature changes, and satellite measurements of sea ice cover. The two applications are related as rising temperatures cause the ice cover to decrease, but the manifolds of interest are very different: $S^2 \times S^1$ and a complicated subset of $S^2$.



\subsection{Application 1}
\paragraph{Data and model}
	Data consists of monthly averages of temperatures, starting in 1983 and ending in 2007, for each $1^\circ \times 1^\circ$ tile on Earth, using standard latitudes and longitudes. Temperatures are satellite measurements collected by NASA.\footnote{These data were obtained from the NASA Langley Research Center Atmospheric Science Data Center Surface meteorological and Solar Energy (SSE) web portal supported by the NASA LaRC POWER Project. 
 \url{http://eosweb.larc.nasa.gov}} 
The aim is to test for a positive increase in temperature, for each point on Earth and each month of the year. Since year has a natural periodic structure, we identify it with $S^1$. Earth is naturally identified with $S^2$, hence the domain for our data is the manifold $S^2 \times S^1$.

In detail, we applied the following linear regression model:
\begin{equation}
    y_i(x,s) = a(x,s) + b(x,s) t_i + \epsilon_i(x,s),
    \quad x \in S^2, s \in S^1, i \in \{1, \dots, 25\}
\end{equation}
where $t_i$ corresponds to years from  $1983$ to 2007.

\paragraph{Hypothesis testing and inference} 
We tested for \emph{positive trend} i.e.
$H_0(x,s): b(x,s) = 0$ with alternative hypothesis $H_A(x,s): b(x,s) > 0$. 
To test this, we used a permutation test to evaluate the hypothesis, using the t statistic with cutoff, ie:
$$
T(x,s) = \max \left(0, \frac{\hat{b}(x,s)}{SE(b(x,s))} \right)
$$
where $\hat{b}$ and $SE(b)$ are the estimate and standard error for $b$, respectively. 

In order to study the trade-off between large correction sets and power, we applied various maximal radii for the ball sizes on Earth (cf. Equation \eqref{iwt-ball}).

\paragraph{Triangulation and preprocessing}
One crucial feature is that data were more densely sampled closer to the poles than close to Equator. This is a bad choice for the triangulation, so we used an icosahedron-based tessellation of the sphere. The tessellation was of order 25. That is, each face of the icosahedron was partioned into $25^2$ triangles, totalling 6252 vertices for the triangulation. 
To get values for this grid, we applied a local linear smoother using the kernel $K(x,y) = max(\tfrac{\pi}{180} - d(x,y), 0)$,
where d is geodesic distance on the sphere, measured in radians.
The distance between two neighbouring points in the triangulation is roughly 300 km. 

For the seasonal cycle on $S^1$, we used a grid of 12 points, one for each month. 

\paragraph{Results} 
To investigate the trade-off between Type I and II errors, we applied the presented methodology with three different radii, 
$\{r_1, r_2, r_3\}$ where $r_1 \approx 1618 \mathrm{km}, r_2 \approx 3653 \mathrm{km}, r_3 \approx 7660 km$.

\begin{figure}
    \centering
    \includegraphics[width=0.8\textwidth]{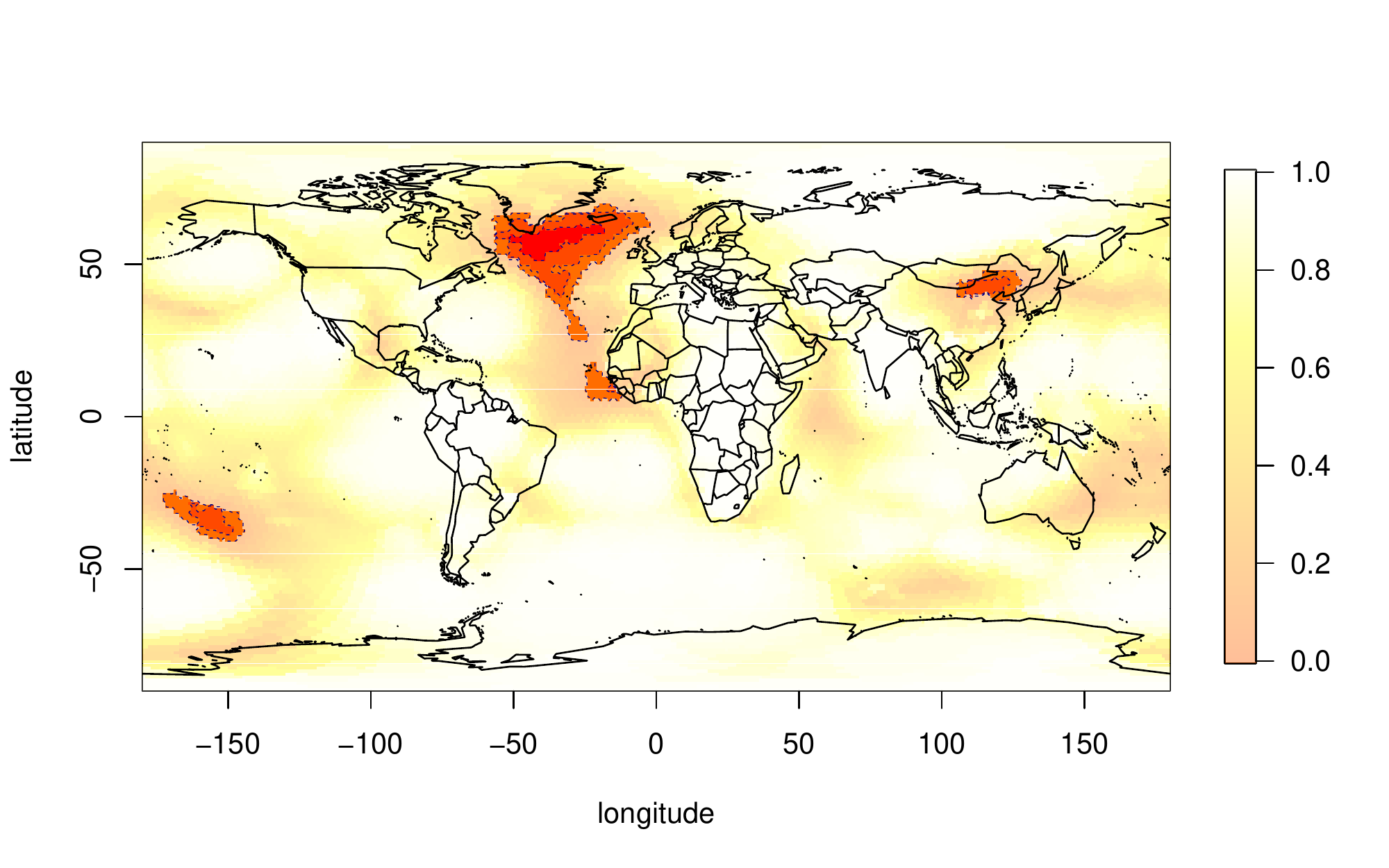}
    \includegraphics[width=0.8\textwidth]{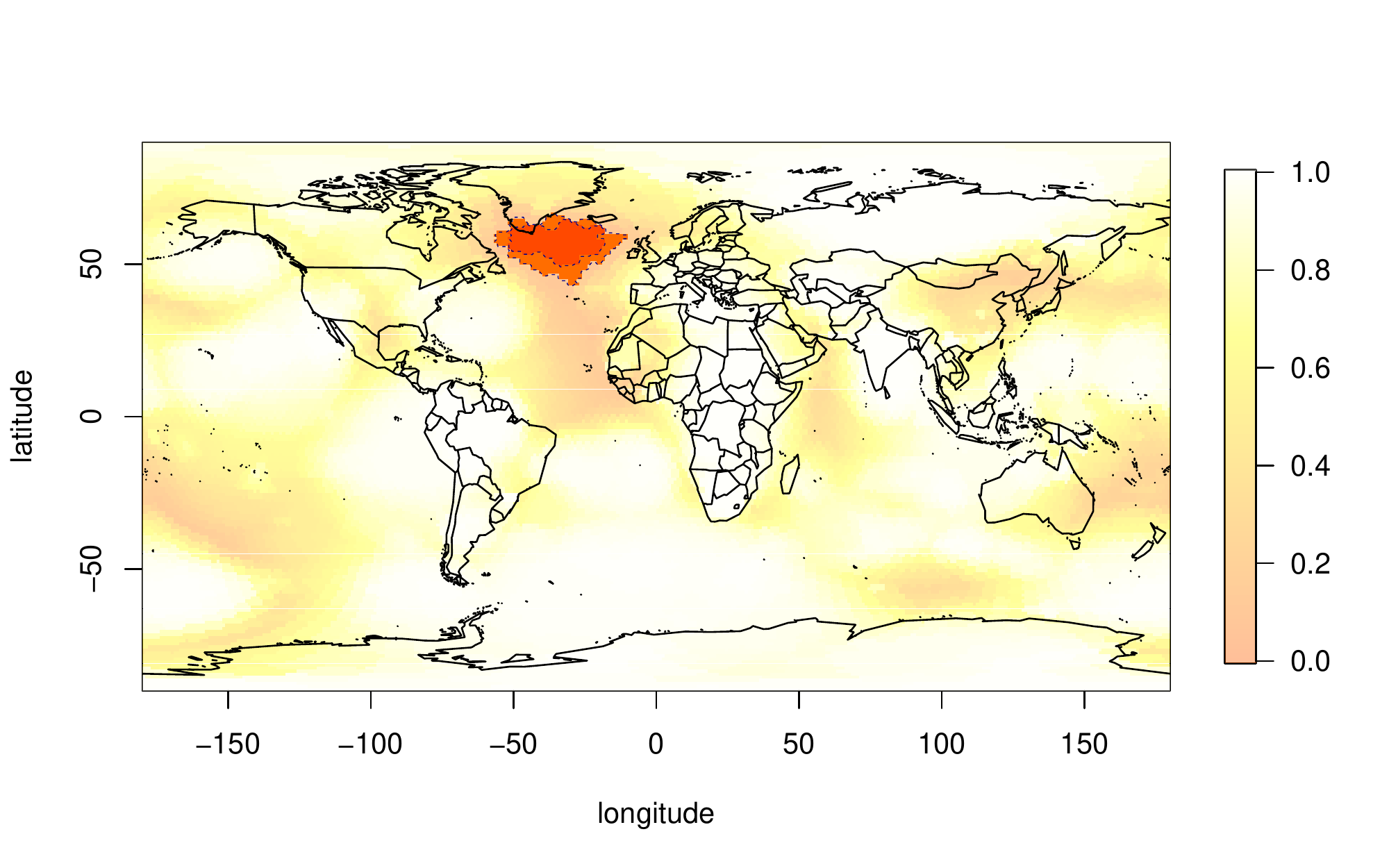}
    \includegraphics[width=0.8\textwidth]{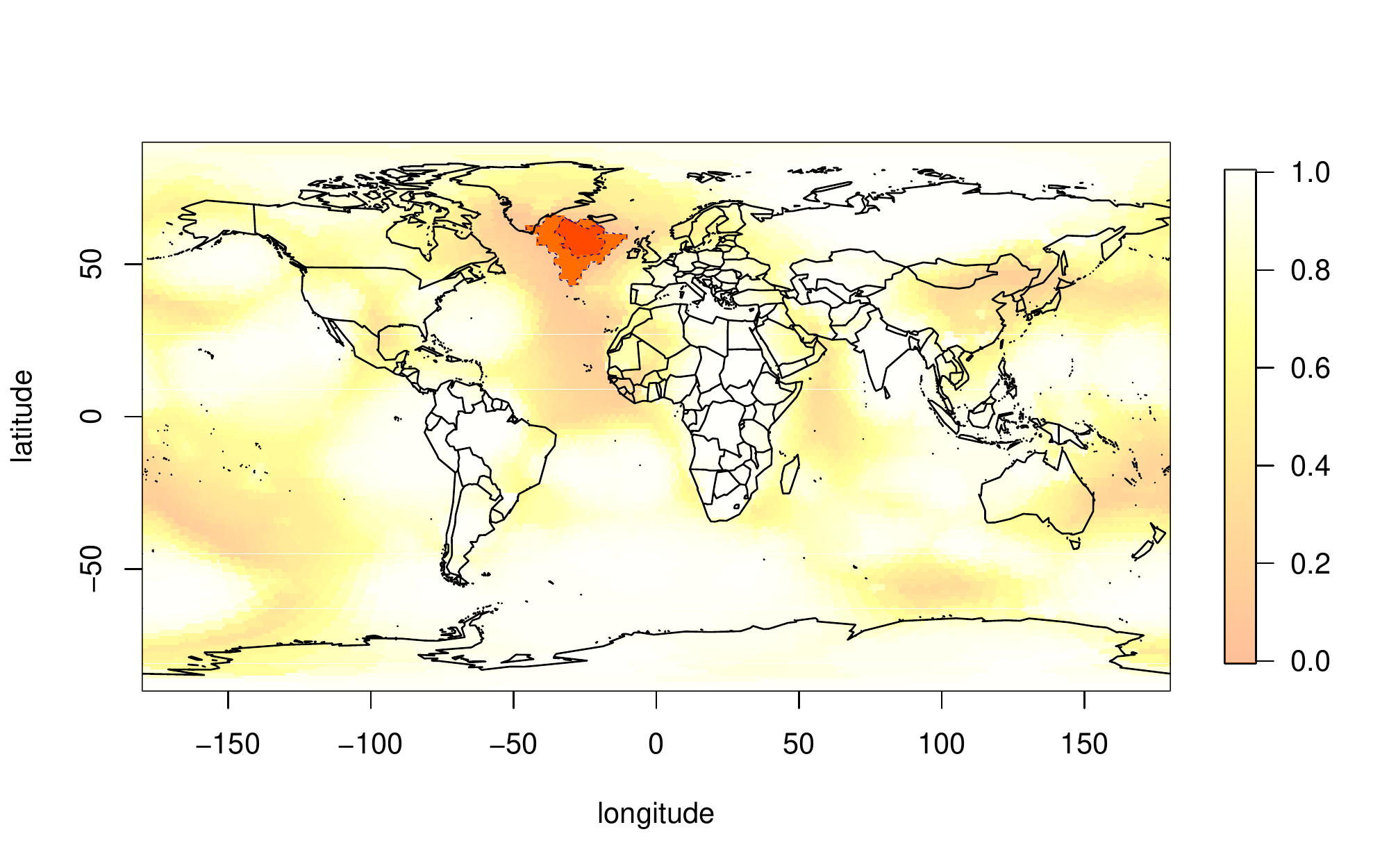}
    \caption{Adjusted $p$ values in September for three different max radii. Upper $r_1 \approx 1618\mathrm{km}$, Middle: $ r_2 \approx 3653\mathrm{km}$ and bottom: $r_3 \approx 7660 km$}
    \label{fig-sep-r3}
\end{figure}
Results for September are shown in Figure \ref{fig-sep-r3}. 
As expected, the regions of significance decrease with increased radii. At our largest scale, only a small region of the Northern Atlantic is significant. 
The results can be interpreted in following way: \textit{we find significant evidence of global warming within every ball of radius $< 7700 $km around $(26W, 56N)$ and every "season" spanning September}.

\begin{figure}
    \centering
    \includegraphics[width=\textwidth]{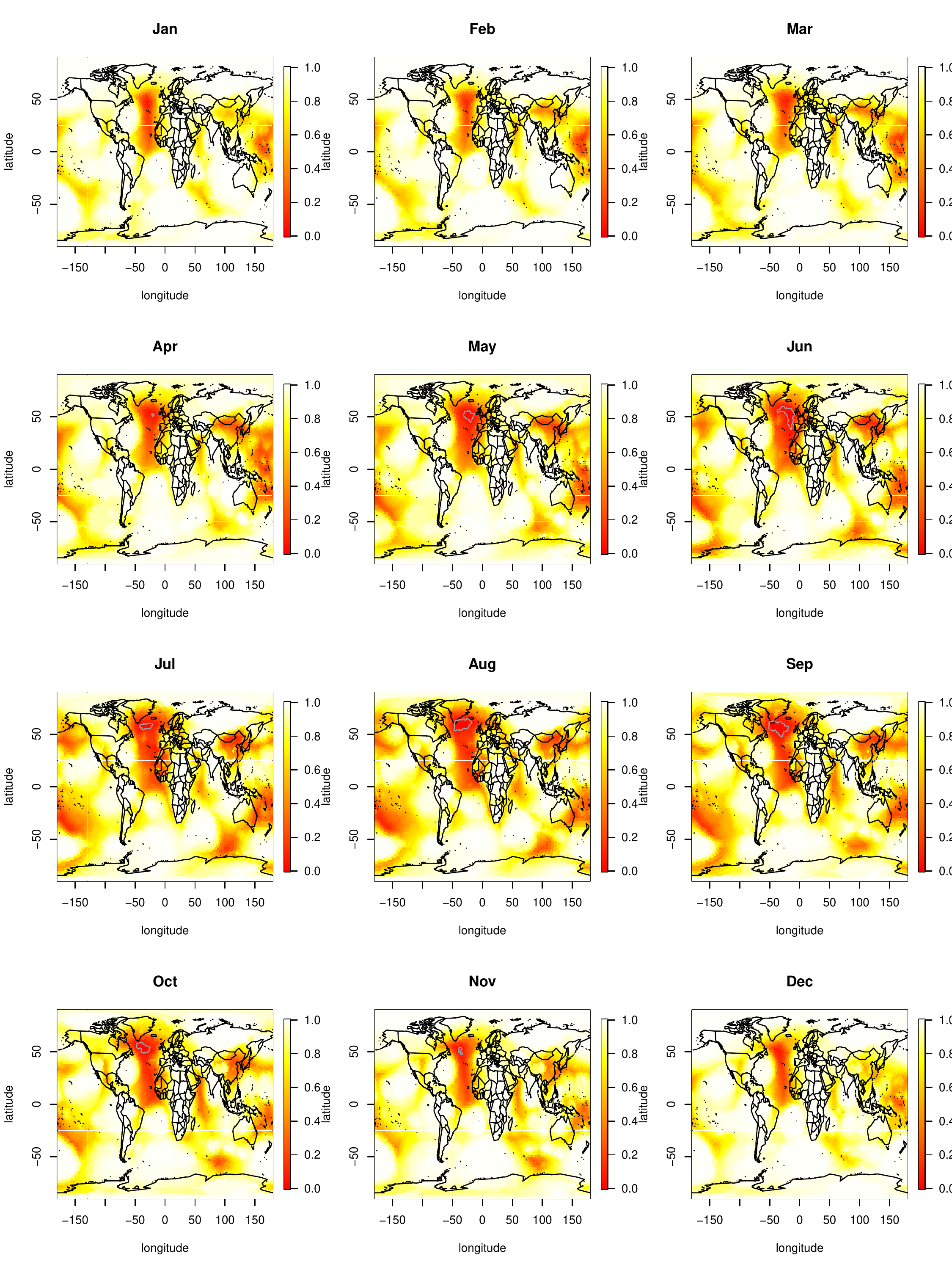}
    \caption{Adjusted $p$ values, $r = r_3 \approx 7660 km$}
    \label{just-pv-aar}
\end{figure}

In Figure \ref{just-pv-aar}, we are considering the different months: Here is the trend clear: we see stronger evidence of global warming in the Northern summer months, and again only the Northern Atlantic has significance.

\subsection{Application 2: Ice cover for the Northern Hemisphere} 
\paragraph{Data and preprocessing}
Data consists of yearly measurements of ice cover on the northern hemisphere, as measured by the satellites of Copernicus Programme\footnote{Copernicus Climate Change Service, Climate Data Store, (2020): Sea ice concentration daily gridded data from 1979 to present derived from satellite observations. Copernicus Climate Change Service (C3S) Climate Data Store (CDS), DOI:10.24381/cds.3cd8b812}. We use the observations from 1 February in the years 1987--2015. Data were given as a percentage between 0 (no ice) and 100 (full ice cover).

We preprocessed the data first by sieving the data 
to a finely spaced grid of 28749 data points with neighbouring distance $\approx 60 km$. 
From this grid a triangulation $G$ was constructed and weights for each point calculated (cf. formula \eqref{weight-eq}).
As a second step, the region of interest was trimmed and all non-sea points removed. This resulted in the map underlying Figures \ref{istrend-figur} and \ref{isplot-figur} and constitutes our manifold $M$. In total, $M$ is approximated by 9695 points. 
To calculate distances between two points in $V$, Dijkstra's algorithm was applied on the original triangulation $G$, but restricted to edges connecting vertices in $V$. This ensured that only sea paths were followed; the longest distance was around 14000 km connecting two points on the east coast of Asia and east coast of North America through a winding route.

\paragraph{Hypothesis testing and inference} 
We applied the following linear regression model:
\begin{equation}
    y_i(x) = a(x) + b(x) t_i + \epsilon_i(x),
    \quad x \in M, i \in \{1, \dots, 29\}
\end{equation}
where $t_i$ corresponds to years from  $1987$ to 2015, 
and tested for changes in trend, $H_0(x): b(x) = 0$ with a two-sided alternative.
We remark that a linear model is not a good choice for modelling ice cover as a function of time, however for the purpose of inference using permutation tests, this is an excellent choice. 

The test statistic used was the squared trend: 
$$
T(x) = \hat{b}^2(x) = \left(\frac{\sum (y_i(x) - \bar{y}(x))(t_i - \bar{t})}{\sum (t_i -\bar{t})^2} \right)^2, \quad x \in M
$$ 

This ensured that regions contributed to the test statistic proportional to how much the ice cover varied (in February, across years). We used 500 permutations in the permutation test.

\paragraph{Results}
\begin{figure}
    \centering
    \includegraphics[width=0.45\textwidth]{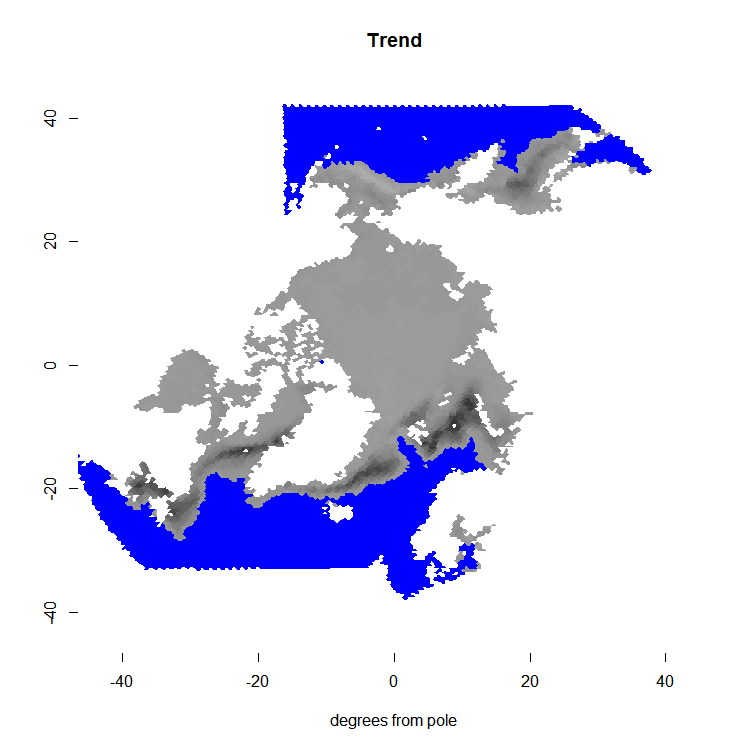}
    \includegraphics[width=0.45\textwidth]{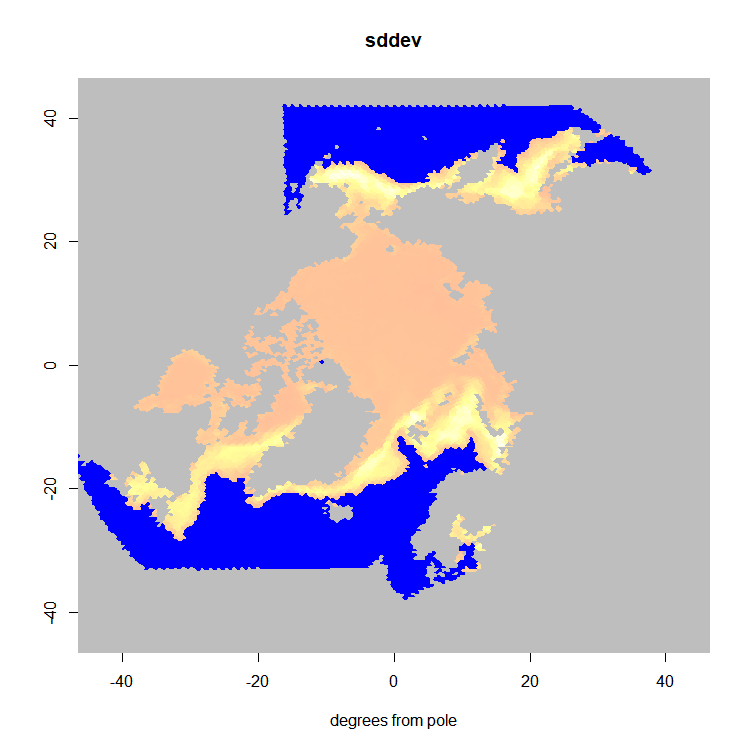} \\[-0.2cm]
    \caption{Pointwise estimates for trend and residual standard deviation, black indicates decrease in cover and orange indicates low standard deviation. Blue are ice-free regions}
    \label{istrend-figur}
\end{figure}
As illustrated in Figure \ref{istrend-figur}, the variation in data is seen along the edges of the arctic, mostly associated with decreasing ice cover.

Maps of unadjusted and adjusted p-values can be seen in Figure \ref{isplot-figur} including an intermediate adjustment with $r_\text{max} = 0.2$ radians ($\approx 1275$km). Note that $p$ values are trivially equal to 1 in the blue areas (ice-free regions). 

\begin{figure}
    \centering
    \includegraphics[width=0.5\textwidth]{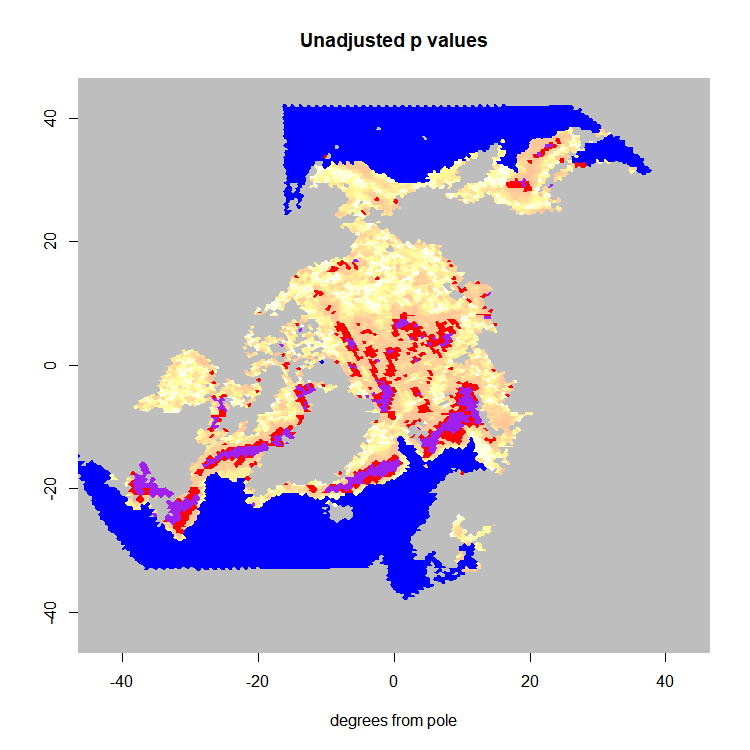} \\[-0.2cm]
    \includegraphics[width=0.5\textwidth]{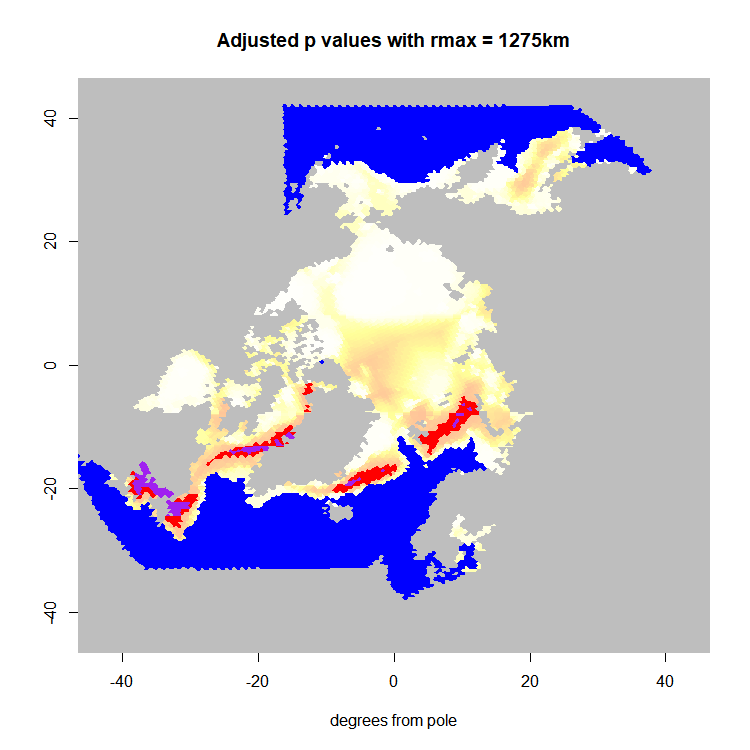} \\[-0.2cm]
    \includegraphics[width=0.5\textwidth]{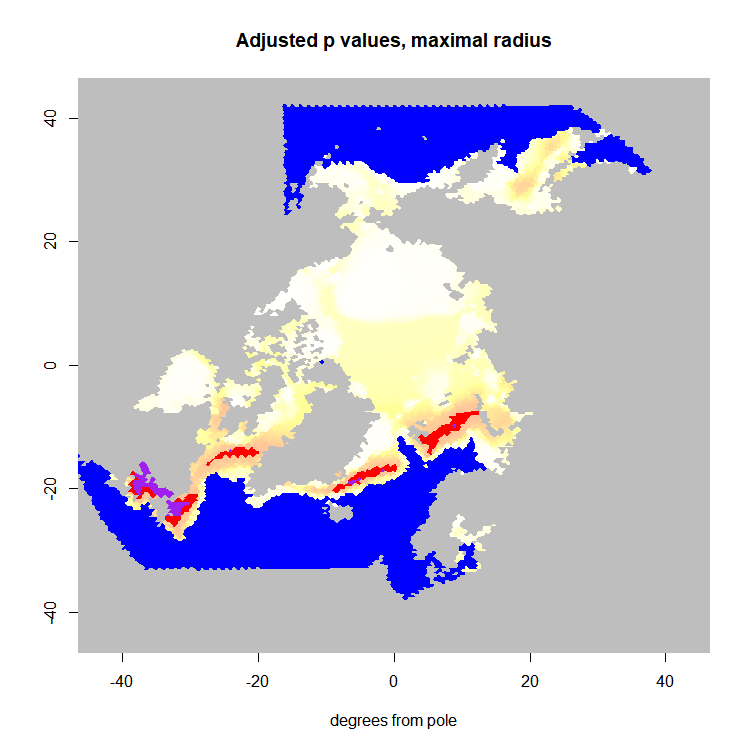} \\
    \caption{Unadjusted and adjusted $p$ values for the ice cover application. Purple denotes areas that are significant on a 1\% level, red areas that are significant on a 5\% level. Blue are ice-free regions}
    \label{isplot-figur}
\end{figure}
We see a rather noisy signal in the unadjusted p-value map, however this soon disappears even with a moderate adjustment as seen in the second figure. At the largest adjustment radius some regions are still significant: we find significant evidence of loss of ice cover in the Barents Sea, Greenland Sea, Baffin Bay and the waters around Newfoundland, even after the strongest adjustment possible. 
We also remark the strong transition from highly significant regions to non-significant regions in some places; this is due to the fact that ice-free regions do not contribute to the test statistic.

\section{Discussion} \label{sect-discuss}
In this paper we have introduced a new methodology for local inference on manifold domains. 
    The methodology is very flexible with its basic assumption being that data are i.i.d. under the null hypothesis. This contrasts for instance with \cite{olsen2021} which has the very technical assumption PDRS that is impossible to verify in practice. We demonstrated the methodology in two data examples and a simulation, where the manifold domains were far from trivial. 

We would also like to stress the flexibility in terms of data and domain. This was demonstrated in the domains used for the simulation and applications. For the second application we also chose some difficult data: each instance was map of highly correlated percentage values, which would be a major challenge to model using parametric methods. 

We formulated our methodology using Riemannian manifolds. %
However despite the naturalness of Riemannian manifolds, we may in principle generalize our notion of local inference to all mathematical objects with non-trivial Hausdorff measure such  as fractals. We would be very curious to see this applied to functional data defined on a fractal domain. 
An interesting discussion point is the maximal radius
that performs trade-off between Type I and Type II errors. In this paper $r$ is selected a priori, but it would of interest to select $r$ using a data-driven approach.

Since the proposed method is based on permutation tests and require evaluating $T^I$ for all discretized $I$, the proposed method has a fairly high computational cost. However it should not be a big obstacle in practice; the procedure was implemented in R and the two applications ran on a typical high-end laptop. 

We identify several possible directions for future research including:
weighted measures, functional covariates, functional-on-scalar
linear models (cf. \cite{abramowiczmox}) and a deeper investigation into  selection of the radius parameter. We leave  this as future work. 
 
 \section*{Acknowledgements}
We are grateful to Associate Professor Simona Perotto (Politecnico di Milano) and ADAPTA Studio\footnote{\url{https://adapta.studio/}} for granting us use of the chameleon in the simulation study. 
 
\bibliographystyle{plain}

\begin{thebibliography}{10}
	
	\bibitem{abramowiczmox}
	Konrad Abramowicz, Charlotte~K H{\"a}ger, Alessia Pini, Lina Schelin, Sara
	Sj{\"o}stedt~de Luna, and Simone Vantini.
	\newblock Nonparametric inference for functional-on-scalar linear models
	applied to knee kinematic hop data after injury of the anterior cruciate
	ligament.
	\newblock {\em Scandinavian Journal of Statistics}, 45(4):1036--1061, 2018.
	
	\bibitem{BH1995}
	Yoav Benjamini and Yosef Hochberg.
	\newblock Controlling the false discovery rate: a practical and powerful
	approach to multiple testing.
	\newblock {\em Journal of the royal statistical society. Series B
		(Methodological)}, pages 289--300, 1995.
	
	\bibitem{freedman1983nonstochastic}
	D.~Freedman and D.~Lane.
	\newblock A nonstochastic interpretation of reported significance levels.
	\newblock {\em J. Bus. Econ. Stat.}, 1(4):292--298, 1983.
	
	\bibitem{liebl2019}
	Dominik Liebl and Matthew Reimherr.
	\newblock Fast and fair simultaneous confidence bands for functional
	parameters.
	\newblock {\em arXiv preprint arXiv:1910.00131}, 2019.
	
	\bibitem{locatelli1}
	L~Locatelli, S.~Perotto, and F.~Clerici.
	\newblock Metodo implementato mediante computer per la rimappatura di una
	texture di un oggetto grafico tridimensionale, Italian patent No.
	102021000018920. Filed on July 16, 2021.
	
	\bibitem{olsen2021}
	Niels~Lundtorp Olsen, Alessia Pini, and Simone Vantini.
	\newblock False discovery rate for functional data.
	\newblock {\em Test}, pages 1--26, 2021.
	
	\bibitem{pataky2021}
	Todd~Colin Pataky, Konrad Abramowicz, Dominik Liebl, Alessia Pini,
	Sara~Sj{\"o}stedt de~Luna, and Lina Schelin.
	\newblock Simultaneous inference for functional data in sports biomechanics.
	\newblock {\em AStA Advances in Statistical Analysis}, pages 1--24, 2021.
	
	\bibitem{perotto1}
	S.~Perotto, L.~Locatelli, M.~Carbonara, and F.~Clerici.
	\newblock Metodo implementato mediante computer per la semplificazione di una
	mesh di un oggetto grafico tridimensionale, Italian patent No.
	102022000001328. Filed on January 01, 2022.
	
	\bibitem{pesarin}
	Fortunato Pesarin and Luigi Salmaso.
	\newblock {\em Permutation tests for complex data}.
	\newblock Wiley, 2010.
	
	\bibitem{pv2017interval}
	Alessia Pini and Simone Vantini.
	\newblock Interval-wise testing for functional data.
	\newblock {\em Journal of Nonparametric Statistics}, 29(2):407--424, 2017.
	
	\bibitem{winkler2014}
	Anderson~M Winkler, Gerard~R Ridgway, Matthew~A Webster, Stephen~M Smith, and
	Thomas~E Nichols.
	\newblock Permutation inference for the general linear model.
	\newblock {\em Neuroimage}, 92:381--397, 2014.
	
\end{thebibliography}

\end{document}